\documentclass[runningheads,fleqn]{llncs}

\usepackage{iftex}

\ifpdf
  \usepackage{underscore}         % Only needed if you use pdflatex.
  \usepackage[T1]{fontenc}        % Recommended with pdflatex
\else
  \usepackage{breakurl}           % Not needed if you use pdflatex only.
\fi

\usepackage{hyperref}
\usepackage[hyphenbreaks]{breakurl}
\usepackage{graphicx}
\usepackage{color}
\usepackage{amsmath,amssymb}
\usepackage{calc,xspace,stmaryrd,times}
\usepackage{procalg}
\usepackage{tikz}
\usetikzlibrary{arrows,shapes,automata,positioning}
\tikzset{
        state/.style={
          circle,
          draw,
          minimum size=6mm,
    },
}
\usepackage{ulem}
\newtheorem{defi}{Definition}
\newtheorem{lem}{Lemma}

\newtheorem{thm}{Theorem}
\newtheorem{cor}{Corollary}
\newtheorem{exa}{Example}

\newcommand{\pref}[2][a]{\ensuremath{{#1}.{#2}}}

\renewcommand{\merge}{\mathrel{\parallel}}

\newcommand{\nbisim}[1][]{%
    \setbox0=\hbox{\kern-.1ex{$\leftrightarrow$}\kern-.1ex}
    \setbox1=\vbox{\hbox{\raise .1ex \box0}\hrule}%
    \ensuremath{\not\mathrel{\hbox{\kern.1ex\box1\kern.1ex}_{#1}}}
  }

\newcommand{\defeqn}{\ensuremath{\mathrel{\stackrel{\textrm{def}}{=}}}}

\title{The Queue Automaton Revisited}
\author{Jos C.M. Baeten\inst{1} \and
Bas Luttik\inst{2}}
\authorrunning{J.C.M. Baeten and B. Luttik}
 \institute{CWI, Amsterdam, The Netherlands \and
 Eindhoven University of Technology,
 Eindhoven, The Netherlands}
 \titlerunning{The Queue Automaton}
\authorrunning{J. C. M. Baeten \& B. Luttik}

\begin{document}
\maketitle

\begin{abstract}
We consider the computational model of the Queue Automaton. An old result is that the deterministic queue automaton is equally expressive as the Turing machine. We introduced the Reactive Turing Machine, enhancing the Turing machine with a notion of interaction. The Reactive Turing Machine defines all executable processes. In this paper, we prove that the non-deterministic queue automaton is equally expressive as the Reactive Turing Machine.
Together with finite automata, pushdown automata and parallel pushdown automata, queue automata form a nice hierarchy of executable processes, with stacks, bags and queues as central elements.
\end{abstract}

\section{Introduction}

Replacing, in a pushdown automaton, the (last-in first-out) stack memory by a (first-in first-out) queue memory yields the computational model of the queue automaton. This computational model, sometimes also called a Post machine or a pullup automaton, has not raised a lot of attention in the literature, but it is a known result that the deterministic queue automaton is equally expressive as the Turing machine of \cite{Tu36}, so that it defines all computable languages and all computable functions. Implicitly, this result is already mentioned in \cite{Pos43}, and further given in \cite{Vol70,Man74}.

In this paper, we investigate the (non-deterministic) queue automaton. We do not define the language of a queue automaton directly, but instead, we define its process graph or transition system. A state of this process graph is given by the state of the queue automaton together with the contents of the queue and a transition is given by the label of the transition of the queue automaton. By considering the language equivalence class of the process graph, we obtain again the language, but we can also divide out other equivalence relations. Notable among these is branching bisimilarity (see \cite{GW96}). By dividing out branching bisimilarity, we obtain the \emph{process} of the queue automaton, incorporating a notion of interaction or communication. These notions of process and communication come from process theory or concurrency theory \cite{Mil89,BBR10}.
We prove that several variants of the queue automaton yield the same set of languages and the same set of processes, and prove that a queue automaton with two queues also yields the same set of languages and the same set of processes.

Also the Turing machine can be extended to incorporate processes and communication, but this is not so straightforward as we just sketched for the queue automaton. We achieved this, nonetheless, in \cite{BLT13}, where we introduced the Reactive Turing Machine. Whereas the classical Turing machine defines the class of computable languages and computable functions, the reactive Turing machine also yields a process graph that can be used to  define the class of executable processes. In this paper, we prove that the non-deterministic queue automaton is equally expressive as the reactive Turing machine. This shows again that the notion of executability we introduced is robust: it is given by different computational models and also by the $\pi$-calculus, see \cite{LY21}.

The queue automaton has certain advantages over the reactive Turing machine: it is mathematically simpler, the extension with interaction is easier, and we get a better process hierarchy, as we explain now. In the queue automaton, we can make the interaction between the finite control and the queue memory explicit, by proving that every executable process is branching bisimilar to a regular process communicating with a queue. In earlier papers \cite{BCT08,BL23}, we proved that every pushdown process is branching bisimilar to a regular process communicating with a stack, and every parallel pushdown process is branching bisimilar to a regular process communicating with a bag. Thus, the queue, stack and bag are the central elements in this Chomsky-Turing hierarchy.

This paper contributes to our ongoing project to integrate the theory
of automata and formal languages on the one hand and concurrency
theory on the other hand. The integration requires a more refined view
on the semantics of automata, grammars and expressions. Instead 
of treating automata as language acceptors, and grammars and
expressions as syntactic means to specify languages, we propose to view
them both as defining process graphs. The great benefit of this
approach is that process graphs can be considered modulo a plethora of
behavioural equivalences \cite{Gla93}. One can still consider language
equivalence and recover the classical theory of automata and formal
languages. But one can also consider finer notions such as
bisimilarity, which is better suited for interacting processes.

\section{Preliminaries}

As a common semantic framework we use the notion of a
\emph{labelled transition system}.

\begin{defi} \label{def:tsspace}
A \emph{labelled transition system} is a quadruple
$(\mathcal{S},\mathcal{A},{\xrightarrow{}},{\downarrow})$, where
\begin{enumerate}
\item $\mathcal{S}$ is a set of \emph{states};
\item $\mathcal{A}$ is a set of \emph{actions}, $\tau\not\in\mathcal{A}$ is the \emph{unobservable} or \emph{silent} action;
\item
${\xrightarrow{}}\subseteq{\mathcal{S}\times(\mathcal{A}\cup\{\tau\})\times\mathcal{S}}$ is
an $\mathcal{A}\cup\{\tau\}$-labelled \emph{transition relation}; and
     \item ${\downarrow}\subseteq\mathcal{S}$ is the set of \emph{final}, \emph{accepting} or \emph{terminating} states.
 \end{enumerate}
A \emph{process graph} is a
labelled transition system with a special
designated \emph{root state} or \emph{initial state} ${\uparrow}$, i.e., it is a quintuple
$(\mathcal{S},\mathcal{A},{\rightarrow},{\uparrow},{\downarrow})$ such that
$(\mathcal{S},\mathcal{A},{\rightarrow},{\downarrow})$ is a labelled transition system, and ${\uparrow}\in\mathcal{S}$.
We write $s\xrightarrow{a}s'$ for $(s,a,s')\in{\rightarrow}$
and $\term{s}$ for $s\in\mathalpha{\downarrow}$. We write $\mathcal{A}_{\tau}$ for $\mathcal{A} \cup \{\tau\}$.
\end{defi}

For $w\in\Act^{*}$ we define
$s\steps{w}t$ inductively, for all states $s,t,u$: first, $s \steps{\varepsilon} s$, and then, for $a \in \Act$, if $s \step{a} t$ and $t \steps{w} u$, then $s \steps{aw} u$, and
if $s \step{\tau} t$ and $t \steps{w} u$, then $s \steps{w} u$.
We see that $\tau$-steps do not contribute to the string $w$.
We write
$s\step{}t$ for there exists $a\in\Act_{\tau}$ such that
$s\step{a}t$. Similarly, we write $s\steps{}t$ for ``there exists
$w\in\Act^{*}$ such that $s\steps{w}t$'' and say that $t$ is
\emph{reachable} from $s$. 
%If $s\steps{w}t$ takes at least one step, we write $s\steps{w}^{+}t$.
%We write $s \not\step{a}$ if there is no $t \in \mathcal{S}$ with $s \step{a} t$.
Finally, we write $s \step{(a)} t$ for ``$s \step{a} t$ or $a = \tau$ and $s = t$''.

By considering language equivalence classes of process graphs, we
recover language equivalence as a semantics, but we can also consider other
equivalence relations. Notable among these is \emph{bisimilarity}.

\begin{defi}
  Let $(\mathcal{S},\mathcal{A},\rightarrow,{\downarrow})$ be a
  labelled transition system. A symmetric binary relation $R$ on
  $\mathcal{S}$ is a \emph{strong bisimulation} if it satisfies the following
  conditions for every $s,t\in\mathcal{S}$ such that $s\mathrel{R} t$ and for all $a\in\mathcal{A}_{\tau}$:
  \begin{enumerate}
    \item if $s\xrightarrow{a}s'$ for some $s'\in\mathcal{S}$,
      then there is a $t'\in\mathcal{S}$ such that
      $t\xrightarrow{a}t'$ and $s'\mathrel{R}t'$; and
    \item if $s{\downarrow}$, then $t{\downarrow}$.
    \end{enumerate}
    If there is a strong bisimulation relating $s$ and $t$ we write $s \bisim t$.
\end{defi}

Sometimes we can use the \emph{strong} version of bisimilarity
defined above, which does not give special treatment to
$\tau$-labelled transitions. In general, when we do give special treatment to $\tau$-labeled transitions, we use (some form of) \emph{branching bisimulation} \cite{GW96}.

\begin{defi}
Let $(\mathcal{S},\mathcal{A},\rightarrow,{\downarrow})$ be a
  labelled transition system. A symmetric binary relation $R$ on
  $\mathcal{S}$ is a \emph{branching bisimulation} if it satisfies the following
  conditions for every $s,t\in\mathcal{S}$ such that $s\mathrel{R} t$ and for all $a\in\mathcal{A}_{\tau}$:
    \begin{enumerate}
    \item if $s\xrightarrow{a}s'$ for some $s'\in\mathcal{S}$,
      then there are states $t',t'' \in\mathcal{S}$ such that
      $t \steps{\varepsilon} t'' \step{(a)} t'$, $s \mathrel{R} t''$ and $s'\mathrel{R}t'$; and
    \item if $s{\downarrow}$, then there is a state $t' \in \mathcal{S}$ such that $t \steps{\varepsilon} t'$ and  $t'{\downarrow}$.
    \end{enumerate}
    If there is a branching bisimulation relating $s$ and $t$, we write $s \bbisim t$.
   If $s \step{\tau} t$ and $s \bbisim t$, we say this $\tau$-step is \emph{inert}.
\end{defi}

%In this article, we use branching bisimilarity with the extra condition of \emph{rootedness}, in order to ensure that the algebra we consider in the sequel is a congruence.
%
%\begin{defi}
%A branching bisimulation $\mathrel{R}$ is \emph{rooted at the pair} $(s,t)$ if $s \bbisim t$ and the initial steps from this pair are strong, i.e.
% if $s\xrightarrow{a}s'$,
%      then there is a $t'\in\mathcal{S}$ such that
%      $t\xrightarrow{a}t'$ and $s'\mathrel{R}t'$; and
%    if $s{\downarrow}$, then $t{\downarrow}$.
%We write $s \rbbisim t$ if there is a branching bisimulation rooted at the pair $s,t$.
%\end{defi}

\begin{thm}
Strong bisimilarity and branching bisimilarity are equivalence relations on labeled transition systems.
\end{thm}
\begin{proof}
See \cite{Bas96} and \cite{GLT09}.
\end{proof}

Now suppose we are given a labelled transition system and two states $s,t$ in this labelled transition system. These states give rise to two process graphs with the given states as initial states. We say
the two process graphs are strongly bisimilar or branching bisimilar if there is a strong bisimulation or a branching bisimulation on the labelled transition system that relates $s$ and $t$.

A \emph{process} is a branching bisimilarity equivalence class of process graphs. We say a process is \emph{regular} if its branching bisimilarity equivalence class contains an element with finitely many states and finitely many transitions.

Finally, we define when a labelled transition system is \emph{deterministic}.

\begin{defi}
A labelled transition system $(\mathcal{S},\mathcal{A},\rightarrow{},{\downarrow})$ is \emph{deterministic} iff for all $s \in \mathcal{S}$ and for all $a \in \Act_{\tau}$ there is at most one $t \in \mathcal{S}$ with $s \step{a} t$.
Moreover, whenever $s \step{\tau} t$, there is no $a \in \mathcal{A}$ and $u \in \mathcal{S}$ with $s \step{a} u$.

A process graph $(\mathcal{S},\mathcal{A},{\rightarrow},{\uparrow},{\downarrow})$ is deterministic if the labelled transition system $(\mathcal{S},\mathcal{A},{\rightarrow},{\downarrow})$ is deterministic.
\end{defi}

\section{Queue Automata}

We define queue automata, prove that their type of transitions can be restricted without losing expressiveness, and prove that having two queues instead of one does not increase their expressiveness.

 \begin{defi}[Queue automaton]
 A \emph{queue automaton} $Q$ is a sixtuple $({\cal S},{\cal A},{\cal D},\rightarrow,\uparrow,\downarrow)$ where:
 \begin{enumerate}
 \item $\cal S$ is a non-empty finite set of states,
 \item $\cal A$ is a non-empty finite action alphabet,
 \item $\cal D$ is a non-empty finite data alphabet,
 \item $\mathalpha{\rightarrow} \subseteq {\cal S} \times  {\cal A}_{\tau} \times ({\cal D} \cup \{\varepsilon,*\}) \times {\cal D}^{*} \times {\cal S}$ is a finite set of \emph{transitions} or \emph{steps},
 \item $\mathalpha{\uparrow} \in {\cal S}$ is the root state or initial state,
 \item $\mathalpha{\downarrow} \subseteq {\cal S}$ is the set of final, accepting or terminating states.
 \end{enumerate}
 If $(s,a,d,\delta,t) \in \mathalpha{\rightarrow}$ with $d \in \cal D$, we write $s \xrightarrow{a[d / \delta]} t$, and this means that the machine, when it is in state $s$ and $d$ is the head element of the queue, 
 can execute action $a$, dequeue this $d$ and enqueue the string $\delta$ and thereby move to state $t$. Likewise, writing $s \xrightarrow{a[\varepsilon / \delta]} t$ means that the machine, when it is in state $s$ and the queue is empty, can execute action $a$, enqueue the string $\delta$ and thereby move to state $t$. Writing  $s \xrightarrow{a[* / \delta]} t$ means that the machine, when it is in state $s$, can execute action $a$,  enqueue $\delta$ and thereby move to state $t$, irrespective of the contents of the queue and without dequeueing anything.
 In steps $s \xrightarrow{\tau[d / \delta]} t$, $s \xrightarrow{\tau[\varepsilon / \delta]} t$ and $s \xrightarrow{\tau[* / \delta]} t$, no action is executed, only the queue is modified.
 The semantic interpretation of the elements of the action alphabet can be left unspecified, but often they
stand for some kind of interaction with the automaton.
   \end{defi}
   
  In definitions of queue automata appearing in the literature, the set of transitions is sometimes defined a little differently, but all of them are equally expressive, in the sense that they all give rise to the same set of languages and the same set of processes. For instance, in \cite{KMW18}, the $a[*/\delta]$-labeled transitions do not occur, but there are two variants of the $a[d/\delta]$-labeled transitions: one where $d$ is dequeued, and one where $d$ is not dequeued. Moreover, instead of using general sequences $\delta$, only singleton sequences or empty sequences are used. We show some of these expressiveness results in the sequel.
   
  The notion of queue automaton is very similar to the classical notion of a pushdown automaton (see, e.g., \cite{HMU07}), only there, a stack is used instead of a queue, and in a stack, push and pop (as enqueue and dequeue are called) occur according to the last-in first-out principle. Therefore, in a pushdown automaton, we do not need the separate $s \xrightarrow{a[*/\delta]} t$ transitions, as they can be replaced by a  $s \xrightarrow{a[\varepsilon/\delta]} t$ transition in combination with  $s \xrightarrow{a[d/\delta d]} t$ transitions for all $d \in {\cal D}$. Further on, we will see that the expressivity of the queue automaton does not change when we omit these transitions, using a recycle operation and working modulo branching bisimulation.
  
  Now we could proceed to define the language of a queue automaton, but instead, we take an intermediate step and first define the process graph of a queue automaton. By considering the language of this process graph, we find the language of the queue automaton again, and by considering its branching bisimulation equivalence class, we obtain the process of the queue automaton.

\begin{defi}
 Let $Q = ({\cal S},{\cal A},{\cal D},\rightarrow,\uparrow,\downarrow)$ be a queue automaton. The \emph{process graph} of $Q$ is defined as follows, for all $\delta \in {\cal D}^{*}$:
 \begin{enumerate}
 \item  the set of states is $\{ (s,\delta) \mid s \in {\cal S}, \delta \in {\cal D}^{*} \}$;
 %, where we can limit this set to those $(s,\delta)$ that are reachable from $(\uparrow,\varepsilon)$;
 \item the set of actions is ${\cal A}$;
 \item The transition relation is generated by the following clauses:
 \begin{itemize}
 \item  if $s \xrightarrow{a[\varepsilon / \delta]} t$ then $(s,\varepsilon) \step{a} (t,\delta)$;
 \item  if $s \xrightarrow{a[d / \delta]} t$ then  $(s,\zeta d) \step{a} (t,\delta \zeta)$, for all $\zeta \in {\cal D}^{*}$;
\item  if $s \xrightarrow{a[* / \delta]} t$ then $(s,\zeta) \step{a} (t,\delta \zeta)$, for all $\zeta \in {\cal D}^{*}$;
 \end{itemize}
 \item the initial state is $(\uparrow,\varepsilon)$; and
 \item $(s,\delta) \downarrow$ if $s \downarrow$.
 \end{enumerate}
\end{defi}

Usually, we consider only those states $(s,\delta)$ that are reachable from the initial state.
According to this definition, a state $(s, \delta)$ can be final also when the queue contents $\delta$ is non-empty. Defining that only states $(s, \varepsilon)$ can be final is more limiting, and yields a smaller set of processes that are the process of a queue automaton. We can still code in a queue automaton that only states of the form $(s, \varepsilon)$ can be final, by only allowing to enter such a state by means of a transition labeled by $a[\varepsilon/\varepsilon]$. This is illustrated in Examples~\ref{exaww} and \ref{queueqa}. For an extensive treatment of termination conditions for pushdown automata and parallel pushdown automata, see \cite{vT11}.

 \begin{defi}
 Let $Q = ({\cal S},{\cal A},{\cal D},\rightarrow,\uparrow,\downarrow)$ be a queue automaton. The language \emph{accepted} by $Q$, ${\cal L}(Q)$, is the language of its transition system, i.e.
 \[  {\cal L}(Q) = \{ w \in {\cal A}^{*} \mid \exists s \in {\cal S} \quad \exists  \delta \in {\cal D}^{*}  \mbox{ such that } s \downarrow \mbox{ and } (\uparrow,\varepsilon) \steps{w} (s,\delta) \}. \]
 The \emph{process} of $Q$, ${\cal P}(Q)$, is the branching bisimulation equivalence class of its transition system, often represented by its minimal element (identifying all branching bisimilar states).
 \end{defi}
 
 \begin{exa} \label{exaww}
 Let us construct a queue automaton for the language $\{ ww \mid w \in \{a,b\}^{*} \}$. This language is not a pushdown language and also not a parallel pushdown language. 
 We use $a,b$ also as data symbols, so ${\cal D} = \{a,b\}$. In the initial state, a string can be read and enqueued. At some point (non-deterministically) it will switch to dequeue elements of the string again by moving to the second state, using the same order (the queue is first-in-first-out). Acceptance or termination takes place when the queue is empty again. See Figure~\ref{fig:ww}. In the figure, we represent the states by circles, the initial state by a small incoming arrow and a final state by a double circle. An arrow that is labeled with multiple labels means that there is such a transition for each of these labels.
\end{exa}
 
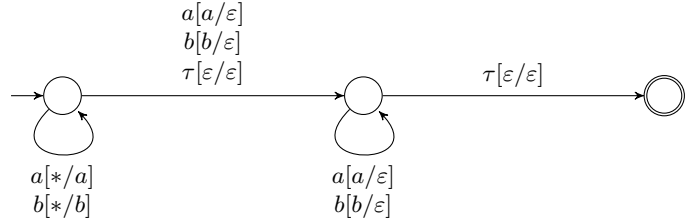
\begin{figure}[htb]
\begin{center}
\begin{tikzpicture}[->,>=stealth',node distance=4cm, node font=\footnotesize, state/.style={circle, draw, minimum size=.5cm,inner sep=0pt}]
  \node[state,initial,initial text={},initial where=left] (s0) {};
  \node[state] [right of=s0] (s1) {};
  \node[state,accepting] [right of=s1] (s2) {};
  
  \path[->]
  (s0) edge[in=315,out=225,loop]
         node[below] {$\begin{array}{c}a[*/a]\\
                                      b[*/b] \end{array}$} (s0)
  (s0) edge node[above] {$\begin{array}{c}a[a/\varepsilon]\\
                                      b[b/\varepsilon] \\
                                      \tau[\varepsilon/\varepsilon] \end{array}$} (s1)
 (s1) edge[in=315,out=225,loop]
         node[below] {$\begin{array}{c}a[a/\varepsilon]\\
                                      b[b/\varepsilon] \end{array}$} (s1)
  (s1) edge node[above] {$\tau[\varepsilon/\varepsilon]$} (s2);                                    
\end{tikzpicture}
\end{center}
\caption{Queue automaton for the language $\{ ww \mid w \in \{a,b\}^{*} \}$.}\label{fig:ww}
\end{figure}

  \begin{exa}
 Figure~\ref{fig:anbncn} shows a queue automaton for the language $\{ a^nb^nc^n \mid n > 0 \}$; it uses data symbols $1,2$. The language is not a pushdown language and also not a parallel pushdown language. 
\end{exa}

\begin{figure}[htb]
\begin{center}
\begin{tikzpicture}[->,>=stealth',node distance=3cm, node font=\footnotesize, state/.style={circle, draw, minimum size=.5cm,inner sep=0pt}]
  \node[state,initial,initial text={},initial where=left] (s0) {};
  \node[state] [right of=s0] (s1) {};
  \node[state] [right of=s1] (s2) {};
  \node[state,accepting] [right of=s2] (s3) {};
  
  \path[->]
  (s0) edge[in=315,out=225,loop]
         node[below] {$a[*/1]$} (s0)
  (s0) edge node[above] {$b[1/2]$} (s1)
 (s1) edge[in=315,out=225,loop]
         node[below] {$b[1/2]$} (s1)
  (s1) edge node[above] {$c[2/\varepsilon]$} (s2)
  (s2) edge[in=315,out=225,loop]
  	node[below] {$c[2/\varepsilon]$} (s2)
(s2) edge node[above] {$\tau[\varepsilon/\varepsilon]$} (s3);                                    
\end{tikzpicture}
\end{center}
\caption{Queue automaton for the language $\{ a^nb^nc^n \mid n>0\}$.}\label{fig:anbncn}
\end{figure}
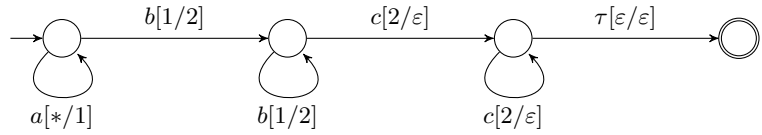
 
 \begin{exa} \label{queueqa}
 We can define the behaviour of the (first-in first-out) queue itself. Given a finite data set $\cal D$, the queue process can execute the following actions:
 \begin{itemize}
 \item $i?d$, {\bf enqueue} data element $d$ (input at port $i$);
 \item $o!d$, {\bf dequeue} data element $d$, if this is the element at the head of the queue (output at port $o$);
 \item $o!\varepsilon$, show that the queue is empty.
 \end{itemize}
 We consider two variants: either the queue can always terminate, or it can terminate only when empty. 
 We can define  queue automata for these two queues, see Figure~\ref{fig:queue}: for all  $d \in \cal D$, there are the edges shown. The queue automaton on the left has only one state and can always terminate, irrespective of the contents of the queue the behaviour of which it represents. The queue automaton on the right has an additional $\tau$-transition, that can only be executed when the queue the behaviour of which it represents has become empty: for this queue automaton, termination can only take place when the queue is empty.
 \end{exa}
 
 \begin{figure}[htb]
\begin{center}
\begin{tikzpicture}[->,>=stealth',node distance=4cm, node
  font=\footnotesize, state/.style={circle, draw, minimum
    size=.75cm,inner sep=0pt}]
  
  \node[state,initial,initial text={},initial where=left,accepting]
  (s0) {};
  \node[state,initial,initial text={}, initial where=left,accepting, right of=s0] (s1) {};
   \node[state, right of=s1] (s2) {};
  
   \path[->]
   (s0) edge[in=315,out=225,loop]
   node[below] {$\begin{array}{c}i?d[*/d] \\
   			o!d[d/\varepsilon] \\
			o!\varepsilon[\varepsilon/\varepsilon]\end{array}$} (s0);
   
  \path[->]
  (s1) edge[in=315,out=225,loop]
  node[below] {$o!\varepsilon[\varepsilon/\varepsilon]$} (s1)
  (s1) edge[bend left] node[above] {$i?d[*/d]$} (s2)
  (s2) edge[in=315,out=225,loop]
  node[below] {$\begin{array}{c} i?d[*/d] \\
  			o!d[d/\varepsilon] \end{array}$} (s2)
   (s2) edge[bend left] node[below]
                                   {$\tau[\varepsilon/\varepsilon]$} (s1);

\end{tikzpicture}
\end{center}
\caption{Queue automata of the queue.}\label{fig:queue}
\end{figure}
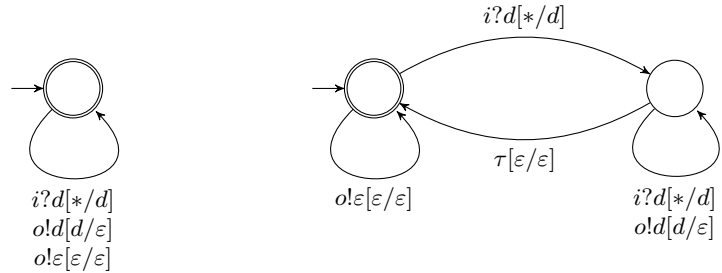

We stated before that the $a[*/\delta]$-labeled transitions are not really needed in the definition of queue automata. We prove this in the following lemma.

\begin{lem} \label{stardollar}
Let  $Q = ({\cal S},{\cal A},{\cal D},\rightarrow,\uparrow,\downarrow)$ be a queue automaton. Then there is a queue automaton $Q'$ that does not have any $a[*/\delta]$-labeled transitions and has the same process as $Q$.
\end{lem}
\begin{proof}
We construct the queue automaton $Q'$ as follows. The data set is $\cal D$ plus an extra data element $\$ \not\in {\cal D}$. The state set is ${\cal S}$ plus fresh states $s^{*}$ for each $s \in {\cal S}$. $Q'$ has the same initial state and the same final states as $Q$, and the same $a[\varepsilon /\delta]$-labeled and $a[d/\delta]$-labeled transitions.
Finally, whenever $Q$ has a transition of the form $s \xrightarrow{a[*/\delta]} t$, $Q'$  has a transition $s \xrightarrow{a[\varepsilon/\delta]} t$ and, for all $d \in {\cal D}$, transitions $s \xrightarrow{a[d/d\$]} s^{*} \xrightarrow{\tau[\$/\delta]} t$ and $s^{*} \xrightarrow{\tau[d/d]} s^{*}$. 

Now we show that the process graphs of $Q$ and $Q'$ are branching bisimilar. We start out from the identity relation on all common states. Note that, whenever $s \xrightarrow{a[*/\delta]} t$ in $Q$, then we have $(s,\varepsilon) \step{a} (t, \delta)$ in the process graph of $Q$ but also in the process graph of $Q'$. For nonempty memory contents, say of the form $\zeta d$ for some $\zeta \in {\cal D}^{*}, d \in {\cal D}$, we have $(s,\zeta d) \step{a} (t, \delta \zeta d)$ in the process graph of $Q$, and $(s,\zeta d) \step{a} (s^{*}, d\$ \zeta) \steps{\varepsilon} (s^{*}, \zeta d \$) \step{\tau} (t, \delta \zeta d)$ in the process graph of $Q'$, so it is enough to relate $s^{*}$ to $t$.
Note that all of the added $\tau$-steps are inert because whenever some state in the process graph has an outgoing $\tau$-transition, then this $\tau$-transition is the unique outgoing transition.
Therefore, the source and target states of the $\tau$-transition are branching bisimilar.
\end{proof}

We use the $a[*/\delta]$-labeled transitions, nonetheless, because they allow for concise descriptions of interesting processes such as the queue of Example~\ref{queueqa}.
In the following proofs, it will be useful on occasion to restrict the transitions in a queue automaton to only singleton enqueues, and to separate enqueues and dequeues.

\begin{defi}
A transition in a queue automaton is a \emph{singleton enqueue} iff it has a label of the form $a[*/d]$ for some $d \in {\cal D}$; it is a \emph{separate dequeue} iff it has a label of the form $a[\varepsilon/\varepsilon]$ or $a[d/\varepsilon]$ for some $d \in {\cal D}$.
\end{defi}

Notice that the queue automata of the queue of Example~\ref{queueqa} have only singleton enqueues and separate dequeues.

\begin{lem} \label{singenqsepdeq}
Let  $Q = ({\cal S},{\cal A},{\cal D},\rightarrow,\uparrow,\downarrow)$ be a queue automaton. Then, there is a queue automaton $Q'$ with only singleton enqueues and separate dequeues with the same process and language.
\end{lem}
\begin{proof}
 Let $N$ be the maximum length of a data sequence occurring in any transition of $Q$, and enumerate the transitions of $Q$ in a sequence $u_1, \ldots, u_K$. The queue automaton $Q'$ uses an extra data element $\$ \not\in {\cal D}$ and has the set of states of $Q$ plus new states $\{s_{ik} \mid i \leq N, k \leq K\} \cup \{s_{1k}^d \mid d \in {\cal D}, k \leq K\}$. $Q'$ has the same action set, initial state and final states as $Q$. The transitions of $Q'$ are defined as follows:
\begin{enumerate}
\item For each transition in $Q$ of the form $s \xrightarrow{a[*/\varepsilon]} t$, if this is transition $u_k$, then $Q'$ has transitions $s \xrightarrow{a[*/\$]} s_{1k} \xrightarrow{\tau[\$/\varepsilon]} t$, and transitions $s_{1k} \xrightarrow{\tau[d/\varepsilon]} s^d_{1k} \xrightarrow{\tau[*/d]} s_{1k}$ for all $d \in {\cal D}$;
\item For each transition in $Q$ of the form $s \xrightarrow{a[*/\delta]} t$ with $\delta \neq \varepsilon$, if this is transition $u_k$ and $\delta = d_1 \ldots d_n$ with $n>1$, then $Q'$ has transitions $s \xrightarrow{a[*/d_n]} s_{1k} \xrightarrow{\tau[*/d_{n-1}]} \cdots s_{(n-1)k} \xrightarrow{\tau[*/d_1]} t$;
\item For each transition in $Q$ of the form $s \xrightarrow{a[\varepsilon/\delta]} t$ with $\delta \neq \varepsilon$, if if this is transition $u_k$ and $\delta = d_1 \ldots d_n$ with $n \geq 1$, then $Q'$ has transitions $s \xrightarrow{a[\varepsilon/\varepsilon]} s_{1k} \xrightarrow{\tau[*/d_{n}]} \cdots s_{nk} \xrightarrow{\tau[*/d_1]} t$;
\item For each transition in $Q$ of the form $s \xrightarrow{a[d/\delta]} t$ with $\delta \neq \varepsilon$, if if this is transition $u_k$ and $\delta = d_1 \ldots d_n$ with $n \geq 1$, then $Q'$ has transitions $s \xrightarrow{a[d/\varepsilon]} s_{1k} \xrightarrow{\tau[*/d_{n}]} \cdots s_{nk} \xrightarrow{\tau[*/d_1]} t$.
\end{enumerate}
In the process graph of $Q'$, all the added $\tau$-steps are inert, so it is branching bisimilar to the process graph of $Q$.
\end{proof}

We see that the notion of the queue automaton is quite robust, as different variants of the queue automaton yield the same set of processes and the same set of languages. As a final illustration of the robustness of the notion of a queue automaton, we show that we can code a memory of two queues into one queue. In order to show this, we first define what a queue automaton with two queues is, and then show that its behaviour can also be obtained by a  queue automaton with one queue.

 \begin{defi}[Queue automaton with two queues]
 A \emph{queue automaton with two queues} $Q$ is a sixtuple $({\cal S},{\cal A},{\cal D},\rightarrow,\uparrow,\downarrow)$ that is just like a queue automaton, only the transition relation $\mathalpha{\rightarrow}$ is now a subset of 
${\cal S} \times  {\cal A}_{\tau} \times ({\cal D} \cup \{\varepsilon,*\})^2 \times ({\cal D}^{*})^2 \times {\cal S}$, i.e. a pair of elements of ${\cal D} \cup \{\varepsilon,*\}$ and a pair of sequences from ${\cal D}^{*}$ is considered.
We write $s \xrightarrow{a[(d,e) / (\delta,\zeta)]} t$ for  $(s, a, d,e, \delta,\zeta, t) \in \mathalpha{\rightarrow}$ (here, $s,t \in {\cal S}, a \in {\cal A}_{\tau}, d,e \in {\cal D} \cup \{\varepsilon,*\}, \delta, \zeta \in {\cal D}^{*}$).
   \end{defi}
   
  From this definition, we get a process graph as expected: the states of the process graph are the triples $(s, \delta, \zeta)$ reachable from initial state $(\uparrow, \varepsilon, \varepsilon)$ by means of the transition relation generated from the following clauses:
  \begin{enumerate}
   \item if $s \xrightarrow{a[(\varepsilon,\varepsilon) / (\delta,\delta')]} t$ then $(s,\varepsilon,\varepsilon) \step{a} (t,\delta,\delta')$;
   \item if $s \xrightarrow{a[(\varepsilon,*) / (\delta,\delta')]} t$ then $(s,\varepsilon,\zeta) \step{a} (t,\delta,\delta' \zeta)$ for all $\zeta \in {\cal D}^{*}$;
   \item if  $s \xrightarrow{a[(*,\varepsilon) / (\delta,\delta')]} t$ then  $(s,\zeta,\varepsilon) \step{a} (t,\delta \zeta,\delta')$ for all $\zeta \in {\cal D}^{*}$;
   \item if  $s \xrightarrow{a[(*,*) / (\delta,\delta')]} t$ then $(s,\zeta,\zeta') \step{a} (t,\delta\zeta,\delta'\zeta')$ for all $\zeta,\zeta' \in {\cal D}^{*}$;
%   \item if and only if  or $s \xrightarrow{a[(*,*) / (\delta,\delta')]} t$ (where $\zeta \neq \varepsilon$);
%   \item if and only if  $s \xrightarrow{a[(*,\varepsilon) / (\delta,\delta')]} t$ or $s \xrightarrow{a[(*,*) / (\delta,\delta')]} t$ (where $\zeta \neq \varepsilon$);
%     \item  if and only if $s \xrightarrow{a[(*,*) / (\delta,\delta')]} t$ (where $\zeta,\zeta' \neq \varepsilon$);
 \item if $s \xrightarrow{a[(d,\varepsilon) / (\delta,\delta')]} t$ then $(s,\zeta d,\varepsilon) \step{a} (t,\delta \zeta,\delta')$ for all $\zeta \in {\cal D}^{*}$;
  \item if $s \xrightarrow{a[(\varepsilon,d) / (\delta,\delta')]} t$ then $(s,\varepsilon,\zeta d) \step{a} (t,\delta,\delta'\zeta)$  for all $\zeta \in {\cal D}^{*}$;
   \item if $s \xrightarrow{a[(d,*) / (\delta,\delta')]} t$ then $(s,\zeta d,\zeta') \step{a} (t,\delta \zeta,\delta'\zeta')$  for all $\zeta,\zeta' \in {\cal D}^{*}$;
  \item if $s \xrightarrow{a[(*,d) / (\delta,\delta')]} t$ then $(s,\zeta,\zeta' d) \step{a} (t,\delta\zeta,\delta'\zeta')$ for all $\zeta,\zeta' \in {\cal D}^{*}$;
   \item if $s \xrightarrow{a[(d,e) / (\delta,\delta')]} t$ then $(s,\zeta d,\zeta' e) \step{a} (t,\delta \zeta,\delta'\zeta')$for all $\zeta,\zeta' \in {\cal D}^{*}$.
\end{enumerate}

\begin{thm} \label{2qtoq}
Let $Q$ be a queue automaton with two queues. Then there is a queue automaton with one queue $M$ such that the process graphs of $Q$ and $M$ are branching bisimilar.
\end{thm}

\begin{proof}
Let $Q$ be a queue automaton with two queues. We define the queue automaton with one queue $M$ in the following. As in the proof of Lemma~\ref{stardollar}, we use an extra data element $\$$ to be able to traverse the contents of a queue. A state in $Q$ will contain the contents of two queues, two sequences $\delta$ and $\zeta$. In $M$, this will be encoded as the contents of one queue: $\delta\between\zeta$, where extra data element $\between$ is the separator between the two sequences.
A state $(s,\delta,\zeta)$ in $Q$ will correspond to the state $(s,\delta\between\zeta)$ in $M$. Enumerate all transitions of $Q$ in a sequence $u_1, \ldots, u_K$. For each transition $u_k$ we add four new states $s^1_k,s^2_k,s^3_k,s^4_k$ in $M$, moreover, there is a new initial state $\uparrow'$ in $M$. The further elements are as follows.
\begin{enumerate}
\item $M$ has the same action set and the same final states as $Q$;
\item $M$ has the data set $\cal D$ of $Q$ and in addition the symbols $\between,\$$;
\item $M$ has initial state $\uparrow'$ and a transition labelled $\tau[\varepsilon / \between]$ to the initial state of $Q$;
\item if transition $s \xrightarrow{a[(*,*) / (\delta,\zeta)]} t$ is transition $u_k$ in $Q$, then there are transitions $s \xrightarrow{a[* /  \$]} s^1_k \xrightarrow{\tau[\between / \between \zeta]} s^2_k \xrightarrow{\tau[\$ / \delta]} t$ in $M$, and for all $d \in {\cal D}$ transitions  $s^1_k \xrightarrow{\tau[d/d]} s^1_k$ and $s^2_k \xrightarrow{\tau[d/d]} s^2_k$; note that in this case, in the process graph of $M$, there is always an $a$-step possible, irrespective of the contents of the queues in the state;
\item if transition $s \xrightarrow{a[(*,\varepsilon) / (\delta,\zeta)]} t$ is transition $u_k$ in $Q$, then there are transitions $s \xrightarrow{a[\between / \between \zeta \$]} s_k^1 \xrightarrow{\tau[\$ / \delta]} t$ and for all $d \in {\cal D}$ transitions  $s_k^1 \xrightarrow{\tau[d / d]} s_k^1$ in $M$; in this case, in the process graph of $M$, there is only an $a$-step possible if the second queue is empty, which can be checked by seeing $\between$ at the head of the queue;
\item if transition $s \xrightarrow{a[(\varepsilon,\varepsilon) / (\delta,\zeta)]} t$  is transition $u_k$ in $Q$, then in $M$, there is only an $a$-labeled transition if both queues are empty. We can check whether the second queue is empty by seeing $\between$ at the head of the queue, for the first queue we have to traverse the queue. The succeeding transitions in $M$ (i.e., the transitions that correspond to the situation that the checks that both queues are empty are successful) are $s \xrightarrow{\tau[\between / \between \$]} s^1_k \xrightarrow{\tau[\$ / \varepsilon]} s^2_k  \xrightarrow{a[\between / \delta\between \zeta]} t$ and the failing transitions (i.e., the transitions that correspond to the situation that the check revealed that one of the queues is not empty) are, for all $d \in {\cal D}$, $s^1_k \xrightarrow{\tau[d/d]} s^3_k$ and $s^3_k \xrightarrow{\tau[d/d]} s^3_k$ and, finally, $s^3_k \xrightarrow{\tau[\$/\varepsilon]} s$;
\item if transition $s \xrightarrow{a[(\varepsilon,*) / (\delta,\zeta)]} t$ is transition $u_k$ in $Q$, then again, we need to check whether the first queue is empty. In $M$, we add the succeeding transitions
 $s \xrightarrow{\tau[*/\$]} s^1_k \xrightarrow{\tau[\between / \varepsilon]} s^2_k \xrightarrow{a[\$ / \delta\between\zeta]} t$ and for all $d \in {\cal D}$ transitions  $s^1_k \xrightarrow{\tau[d / d]} s^1_k$; moreover, we add in $M$ the failing transitions $s^3_k \xrightarrow{\tau[\$ / \varepsilon]} s$ and for all $d \in {\cal D}$ transitions  $s^2_k \xrightarrow{\tau[d/d \between]} s^3_k$ and $s^3_k \xrightarrow{\tau[d / d]} s^3_k$;
  \item if transition $s \xrightarrow{a[(*,d) / (\delta,\zeta)]} t$ is transition $u_k$ in $Q$, then there are transitions $s \xrightarrow{a[d/\$]} s_k^1 \xrightarrow{\tau[\between / \between \zeta]} s_k^2 \xrightarrow{\tau[\$ / \delta]} t$ and for all $f \in {\cal D}$ transitions   $s_k^1 \xrightarrow{\tau[f/ f]} s_k^1$ and $s_k^2 \xrightarrow{\tau[f / f]} s_k^2$ in $M$;
\item if transition $s \xrightarrow{a[(\varepsilon,d) / (\delta,\zeta)]} t$ is transition $u_k$ in $Q$, then, in $M$, then there are succeeding transitions $s \xrightarrow{\tau[d/\$]} s_k^1 \xrightarrow{\tau[\between / \varepsilon]} s_k^2 \xrightarrow{a[\$ / \delta \between \zeta]} t$ and for all $f \in {\cal D}$ transitions  $s_k^1 \xrightarrow{\tau[f / f]} s_k^1$; moreover, there are failing transitions $s^3_k \xrightarrow{\tau[\$ / \varepsilon]} s$ and, for all $f \in {\cal D}$, $s^2_k \xrightarrow{\tau[f / f \between]} s^3_k$ and $s^3_k \xrightarrow{\tau[f/f]} s^3_k$;
\item if transition $s \xrightarrow{a[(d,\varepsilon) / (\delta,\zeta)]} t$ is transition $u_k$ in $Q$, then, in $M$, there are succeeding transitions $s \xrightarrow{\tau[\between / \$]} s_k^1 \xrightarrow{a[d/\between \zeta]} s_k^2 \xrightarrow{\tau[\$ / \delta]} t$ and for all $f \in {\cal D}$ transitions  $s_k^2 \xrightarrow{\tau[f/ f]} s_k^2$; moreover, there are failing transitions $s^1_k \xrightarrow{\tau[\$/\varepsilon]} s$ and $s^3_k \xrightarrow{\tau[\$/\varepsilon]} s$ and, for all $e \in {\cal D}, e \neq d$, $s^1_k \xrightarrow{\tau[e/e\between]} s^3_k$ and, for all $f \in {\cal D}$, $s^3_k \xrightarrow{\tau[f/f]} s^3_k$;
\item if transition $s \xrightarrow{a[(d,*) / (\delta,\zeta)]} t$ is transition $u_k$ in $Q$, then, in $M$, there are succeeding transitions $s \xrightarrow{\tau[*/\$]} s_k^1 \xrightarrow{\tau[\between / \varepsilon]} s_k^2 \xrightarrow{a[d/\between \zeta]} s_k^3 \xrightarrow{\tau[\$ / \delta]} t$ and for all $f \in {\cal D}$ transitions  $s_k^1 \xrightarrow{\tau[f/ f]} s_k^1$ and $s_k^3 \xrightarrow{\tau[f/ f]} s_k^3$; moreover, there are failing transitions $s^1_k \xrightarrow{\tau[\$/\varepsilon]} s$ and $s^4_k \xrightarrow{\tau[\$/\varepsilon]} s$ and, for all $e \in {\cal D}, e \neq d$, $s^2_k \xrightarrow{\tau[e/e\between]} s^4_k$ and, for all $f \in {\cal D}$, $s^4_k \xrightarrow{\tau[f/f]} s^4_k$;
\item if transition $s \xrightarrow{a[(d,e) / (\delta,\zeta)]} t$ is transition $u_k$ in $Q$, then, in $M$, there are succeeding transitions $s \xrightarrow{\tau[e/\$]} s_k^1 \xrightarrow{\tau[\between / \varepsilon ]} s_k^2 \xrightarrow{a[d/\between \zeta]} s_k^3 \xrightarrow{\tau[\$ / \delta]} t$ and for all $f \in {\cal D}$ transitions  $s_k^1 \xrightarrow{\tau[f/ f]} s_k^1$ and $s_k^3 \xrightarrow{\tau[f/ f]} s_k^3$; moreover, there are failing transitions $s^2_k \xrightarrow{\tau[\$/\varepsilon]} s$ and $s^4_k \xrightarrow{\tau[\$/\varepsilon]} s$ and, for all $g \in {\cal D}, g \neq d$, $s^2_k \xrightarrow{\tau[g/g\between]} s^4_k$ and, for all $f \in {\cal D}$, $s^4_k \xrightarrow{\tau[f/f]} s^4_k$.
\end{enumerate}
Finally, it is straightforward to check that in the process graph of $M$, all the extra $\tau$-steps are inert.
\end{proof}

We see that the set of processes and the set of languages given by a queue automaton is not enlarged by adding another queue memory. Further on, we use this to show that the the composition of two interacting processes both given by queue automata is again given by a queue automaton.

 \section{Comparison with Reactive Turing Machines} \label{rtm}

We use the definition of the Reactive Turing Machine (RTM) from \cite{BLT13}.

 \begin{defi}[Reactive Turing machine]
 A \emph{reactive Turing machine} $M$ is a sixtuple $({\cal S},{\cal A},{\cal D},\rightarrow,\uparrow,\downarrow)$ where:
 \begin{enumerate}
 \item $\cal S$ is a non-empty finite set of states;
 \item $\cal A$ is a non-empty finite action alphabet;
 \item $\cal D$ is a non-empty finite data alphabet;
 \item $\mathalpha{\rightarrow} \subseteq {\cal S} \times  {\cal A}_{\tau} \times ({\cal D} \cup \{\Box\}) \times ({\cal D} \cup \{\Box\}) \times \{L,R\} \times {\cal S}$ is a finite set of \emph{transitions} or \emph{steps};
 \item $\mathalpha{\uparrow} \in {\cal S}$ is the initial state;
 \item $\mathalpha{\downarrow} \subseteq {\cal S}$ is the set of final states.
 \end{enumerate}
 The blank $\Box$ represents an empty tape cell. Henceforth, we will write ${\cal D}_{\Box}$ instead of ${\cal D} \cup \{\Box\}$.
 If $(s,a,d,e,T,t) \in \mathalpha{\rightarrow}$, we write $s \xrightarrow{a[d / e]T} t$, and this means that the machine, when it is in state $s$ and $d$ is the data element read by the tape head, 
 can execute action $a$, replace $d$ by $e$, can move one position left ($L$) or right ($R$) and end up in state $t$. 
 %If $a = \tau$, the action is unobservable.
 \end{defi}
 
 It requires quite some notational overhead to define the transition relation and the process graph associated to an RTM. The states of the process graph are the configurations of the RTM, consisting of a state of the RTM, the contents of the tape, and the position of the read/write head on the tape. We represent the tape contents by an element of $({\cal D}_{\Box})^{*}$, replacing exactly one occurrence of a tape symbol $d$ by a \emph{marked} symbol $\check{d}$, indicating that the read/write head is on this symbol. We denote by $\check{\cal D}_{\Box} = \{\check{d} \mid d \in {\cal D}_{\Box}\}$ the set of marked tape symbols; a \emph{tape instance} is a sequence $\delta \in ({\cal D}_{\Box} \cup \check{\cal D}_{\Box})^{*}$ containing exactly one element of $\check{\cal D}_{\Box}$. Note that we do not use $\delta$ exclusively for tape instances; we also use $\delta$ for sequences over ${\cal D}$. A tape instance thus is a finite sequence of symbols that represents the contents of a two-way infinite tape. Henceforth, we do not distinguish between tape instances that are equal modulo the addition or removal of extra occurrences of the blank symbol $\Box$ at the left or right extremes of the sequence. That is, we do not distinguish tape instances $\delta$ and $\zeta$ if ${\Box}^{\omega}\delta{\Box}^{\omega} = {\Box}^{\omega}\zeta{\Box}^{\omega}$.
 
%\begin{defi}
A \emph{configuration} of an RTM ${M}$ is a pair $(s,\delta)$ where $s \in {\cal S}$ is a state of the RTM and $\delta$ is a tape instance.
   %\end{defi}
   
 We define an ${\cal A}_{\tau}$-labelled transition system for each RTM such that a transition  $s \xrightarrow{a[d / e]T} t$ corresponds to a transition $(s,\delta) \step{a} (t,\zeta)$, where in $\delta$ some occurrence of $d$ is marked, and in $\zeta$ this marked $d$ is replaced by $e$, and the symbol to the left in $\zeta$ is marked (if $T = L$) or the symbol to the right in $\zeta$ is marked (if $T = R$). If necessary, a blank $\Box$ is added.
 For this, we use the following notation: if $\delta \in {\cal D}_{\Box}^{*}$, then ${\delta}^{<} = \check{\Box}$ if $\delta = \varepsilon$ and ${\delta}^{<} = \zeta \check{d}$ if $\delta = \zeta d$ for some $d \in {\cal D}_{\Box}, \zeta \in {\cal D}_{\Box}^{*}$. Likewise, if $\delta \in {\cal D}_{\Box}^{*}$, then ${}^{>}{\delta} = \check{\Box}$ if $\delta = \varepsilon$ and ${}^{>}{\delta} = \check{d} \zeta$ if $\delta = d \zeta$ for some $d \in {\cal D}_{\Box}, \zeta \in {\cal D}_{\Box}^{*}$.
 
 \begin{defi}
 Let ${M} = ({\cal S},{\cal A},{\cal D},\rightarrow,\uparrow,\downarrow)$ be an RTM. The \emph{process graph} \emph{associated with} ${M}$ is defined as follows:
 \begin{enumerate}
 \item the set of states is the set of configurations $(s,\delta)$ of $M$;
 \item the set of actions is $\cal A$;
 \item the transition relation $\mathalpha{\rightarrow}$ is the least relation satisfying, for all $a \in {\cal A}_{\tau}, d,e \in {\cal D}_{\Box}, \delta, \zeta \in {\cal D}_{\Box}^{*}$:
 \[ (s, \delta \check{d} \zeta) \step{a} (t, \delta^{<} e \zeta) \Longleftrightarrow s \xrightarrow{a[d/e]L} t \]
 and
 \[ (s, \delta \check{d} \zeta) \step{a} (t, \delta e^{>} \zeta) \Longleftrightarrow s \xrightarrow{a[d/e]R} t\]
 \item the initial state is $(\uparrow, \check{\Box})$;
 \item the set of final states is $\{(s,\delta) \mid s \downarrow\}$.
 \end{enumerate}
 A process is \emph{executable} iff its branching bisimulation equivalence class of process graphs contains a process graph associated with an RTM. A language is \emph{computable} iff it is the language of a process graph associated with an RTM. 
 \end{defi}
 
 \begin{thm}
 A process is executable if and only if it is the process of a queue automaton.
 \end{thm}
 \begin{proof}
 Let ${M} = ({\cal S},{\cal A},{\cal D},\rightarrow,\uparrow,\downarrow)$ be an RTM, and suppose $M$ has $n$ transitions $T_1, \ldots, T_n$. For each transition $T_i$, we have 5 new states $s_i^1, s_i^2, s_i^3, s_i^4, s_i^5$.
 In addition, we have an extra state $\uparrow'$.
 A configuration $(s, \delta \check{d} \zeta)$ of the RTM will correspond to the state $(s, \zeta^R \between \delta d)$ of the queue automaton to be constructed, where the symbol $\between$ is a separator and $\zeta^R$ is the reverse of the string $\zeta$. Note that we have to treat $\Box$ as an extra data element, because the RTM uses blanks in this way.
 Now we define the queue automaton $Q$ as follows:
 \begin{enumerate}
 \item the set of states is $\cal S$, extended with new states $\{s_i^1, s_i^2, s_i^3, s_i^4, s_i^5 \mid i \leq n\} \cup \{\uparrow'\}$;
 \item the set of actions is $\cal A$;
 \item the set of data is ${\cal D} \cup \{\Box,\between,\$\}$ ($\$$ is a new data element used as a bookmark);
 \item if $s \xrightarrow{a[d/e]L} t$ is the transition $T_i$ in $M$, then there are transitions $s \xrightarrow{a[d/\$]} s^1_i \xrightarrow{\tau[\between/e \between \Box]} s_i^2 \xrightarrow{\tau[\$ / \varepsilon]} t$ and, for each $f \in {\cal D}_{\Box}$, transitions $s^1_i \xrightarrow{\tau[f/f]} s^1_i$ and $s^2_i \xrightarrow{\tau[f/f]} s^2_i$ in $Q$; moreover, if $s \xrightarrow{a[d/e]R} t$ is transition $T_j$ in $M$, 
 then there are transitions $s \xrightarrow{a[d/\$ e]} s^1_j \xrightarrow{\tau[\between/\between]} s^2_j  \xrightarrow{\tau[\$ / \Box \$]} s^4_j \xrightarrow{\tau[\$/\varepsilon]} t$ and, for each $f,g \in {\cal D}_{\Box}$, transitions
 $s^2_j \xrightarrow{\tau[g/\varepsilon]} s^3_j \xrightarrow{\tau[\$ / g\$]} s^4_j$, $s^1_j \xrightarrow{\tau[f/f]} s^1_j$, $s^3_i \xrightarrow{\tau[f/f]} s^3_i$ and $s^4_j \xrightarrow{\tau[f/f]} s^4_j$ in $Q$.
 \item the initial state is $\uparrow$ and add a transition $\uparrow \xrightarrow{\tau[\varepsilon/\between \Box]} \uparrow'$ in $Q$.
 \item the set of final states is $\downarrow$.
 \end{enumerate}
 Then, we can establish that the process graph of $Q$ is branching bisimilar to the process graph associated with $M$.
 
 For the other direction, suppose a process has a process graph given by a queue automaton ${Q} = ({\cal S},{\cal A},{\cal D},\rightarrow,\uparrow,\downarrow)$. Without loss of generality, we can suppose $Q$ has only singleton enqueues and separate dequeues. Suppose $Q$ has $n$ transitions $T_1, \ldots, T_n$. For each transition $T_i$, we have a new state $s'_i$.
 A state $(s, \delta d)$ of the queue automaton will correspond to the configuration $(s, \Box \delta \check{d} \Box)$ of the RTM to be constructed, and state $(s, \varepsilon)$ to configuration $(s, \check{\Box})$.
 Now define an RTM $M$ as follows:
 \begin{enumerate}
 \item the set of states is $\cal S$, extended with new states $\{s'_i,s''_i \mid i \leq n\}$;
 \item the set of actions is $\cal A$, the set of data is $\cal D$;
 \item if $s \xrightarrow{a[d/\varepsilon]} t$ is a transition in $Q$, then $s \xrightarrow{a[d/ \Box] L} t$ is a transition in $M$; further, if $s \xrightarrow{a[\varepsilon / \varepsilon]} t$ is a transition in $Q$, then $s \xrightarrow{a[\Box  / \Box]L} t$ is a transition in $M$; lastly, if  $s \xrightarrow{a[* / d]} t$ is the transition $T_i$ in $Q$, then $s \xrightarrow{a[\Box / d]R} s'_i \xrightarrow{\tau[\Box / \Box]L} t$ are transitions in $M$ and, for each $e \in {\cal D}$, $s \xrightarrow{a[e/e]L} s''_i$ and $s''_i \xrightarrow{\tau[e/e]L} s''_i$ and $s'_i \xrightarrow{\tau[e/e]R} s'_i$are also transitions in $M$; finally, also $s''_i \xrightarrow{\tau[\Box / d]R} s'_i$ is a transition in $M$;
 \item the initial state is $\uparrow$, and the set of final states is $\downarrow$.
 \end{enumerate}
 Again, we can establish that the process graph of $M$ is branching bisimilar to the process graph of $Q$.
 \end{proof}
 
 \begin{cor}
 A process is executable if and only if it is the process of a queue automaton. A language is computable if and only if it is the language of a queue automaton.
 \end{cor}
 
 A Turing machine can also be used to define when a function is computable. In \cite{BLT13}, we defined this also for the Reactive Turing Machine. Here, we do this for the queue automaton, and give a couple of examples. We designate an input port $i$ and an output port $o$.
 
 \begin{defi}
We say a queue automaton performs a \emph{computation}\index{computation} if
\begin{enumerate}
\item the process graph of the queue automaton is deterministic;
\item every string in the language of the queue automaton consists of a sequence of inputs $i?d_1 \cdots i?d_n$ and a sequence of outputs $o!e_1 \cdots o!e_m$ interleaved (so not necessarily all outputs after all inputs) for some $n,m\geq 0, d_i, e_j \in \cal D$. Note that the length of input and output may differ.
\end{enumerate}
 
In this case, we say the queue automaton \emph{computes} the function $f$ on a domain of data strings $D$ ($D \subseteq {\cal D}^{*}$) if for all input $w \in D$ it has output $f(w) \in {\cal D}^{*}$. 
 \end{defi}
 
 \begin{exa}
 By an adaptation of Example~\ref{exaww}, we can define the function $f(w) = ww$ on $\mathcal{D}^{*}$. See Figure~\ref{fig:fwww}.
 \end{exa}
 
 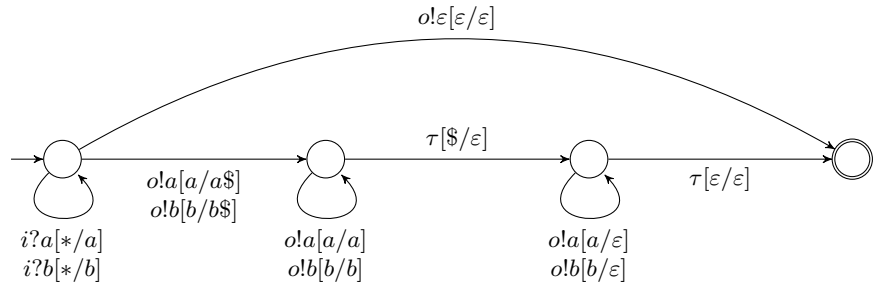
\begin{figure}[htb]
\begin{center}
\begin{tikzpicture}[->,>=stealth',node distance=3.5cm, node font=\footnotesize, state/.style={circle, draw, minimum size=.5cm,inner sep=0pt}]
  \node[state,initial,initial text={},initial where=left] (s0) {};
  \node[state] [right of=s0] (s1) {};
  \node[state] [right of=s1] (s2) {};
  \node[state,accepting] [right of=s2] (s3) {};
  
  \path[->]
  (s0) edge[in=315,out=225,loop]
         node[below] {$\begin{array}{c}i?a[*/a]\\
                                      i?b[*/b] \end{array}$} (s0)
  (s0) edge node[below] {$\begin{array}{c}o!a[a/a\$]\\
                                      o!b[b/b\$]  \end{array}$} (s1)
 (s1) edge[in=315,out=225,loop]
         node[below] {$\begin{array}{c}o!a[a/a]\\
                                      o!b[b/b] \end{array}$} (s1)
  (s1) edge node[above] {$\tau[\$/\varepsilon]$} (s2)
  (s2) edge[in=315,out=225,loop]
  	node[below] {$\begin{array}{c}o!a[a/\varepsilon]\\
				o!b[b/\varepsilon] \end{array}$} (s2)
(s2) edge node[below] {$\tau[\varepsilon/\varepsilon]$} (s3);         

\path[->]
(s0) edge[bend left] node[above] {$o!\varepsilon[\varepsilon/\varepsilon]$} (s3);                           
\end{tikzpicture}
\end{center}
\caption{Queue automaton for the function $f(w) = ww$ on $\{a,b\}^{*}$.}\label{fig:fwww}
\end{figure}

\begin{exa}
There is a queue automaton that compares quantities. Suppose we have two numbers in binary notation and we want to know whether the first number is larger than the second. For simplicity, we assume the numbers have an equal number of digits, adding leading zeroes if necessary. The input consists of the two numbers separated by a $>$-sign, and ${\cal D} = \{0,1,>,\mathit{yes},\mathit{no} \}$.  Figure~\ref{fig:comp} explains the rest. Note that no $\tau$-steps are needed.

Thus, a queue automaton can be used to program a conditional branching in a program. In the same way, we can find queue automata for other program constructs, and other mathematical functions. Notice that there is a clear separation between input and output on the one hand, and memory use on the other hand. In this case, we only need to store the first number in memory, not the second one. Also notice that as soon as we have a difference between the two numbers, we can determine the output, and there is no need for the rest of the input.
\end{exa}

 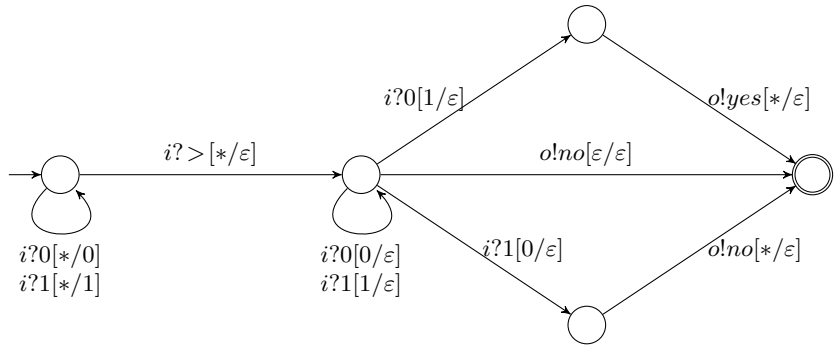
\begin{figure}[htb]
\begin{center}
\begin{tikzpicture}[->,>=stealth',node distance=4cm, node font=\footnotesize, state/.style={circle, draw, minimum size=.5cm,inner sep=0pt}]
  \node[state,initial,initial text={},initial where=left] at (1,2) (s0) {};
  \node[state] at (5,2) (s1) {};
  \node[state] at (8,0) (s2) {};
  \node[state] at (8,4) (s4) {};
  \node[state,accepting] at (11,2) (s3) {};
  
  \path[->]
  (s1) edge node[above] {$o!no[\varepsilon/\varepsilon]$} (s3);
  
  \path[->]
  (s1) edge node[left] {$i?0[1/\varepsilon]$} (s4)
  (s4) edge node[right] {$o!yes[*/\varepsilon]$} (s3);
  
  \path[->]
  (s0) edge[in=315,out=225,loop]
         node[below] {$\begin{array}{c}i?0[*/0]\\
                                      i?1[*/1] \end{array}$} (s0)
  (s0) edge node[above] {$i?\!>\![*/\varepsilon]$} (s1)
 (s1) edge[in=315,out=225,loop]
         node[below] {$\begin{array}{c}i?0[0/\varepsilon]\\
                                      i?1[1/\varepsilon] \end{array}$} (s1)
  (s1) edge node[right] {$i?1[0/\varepsilon]$} (s2)
 (s2) edge node[right] {$o!no[*/\varepsilon]$} (s3);         

\end{tikzpicture}
\end{center}
\caption{Queue automaton comparing quantities.}\label{fig:comp}
\end{figure}

In the following, we use some basic recursion theory. The reader is refered to, e.g., \cite{Rog67}.
In \cite{BLT13}, we characterised the set of executable processes as follows. We call a process graph \emph{effective} if its transition relation and its set of final states are recursively enumerable (with some suitable encoding of these into natural numbers), see \cite{Bou85}. A process is effective if its branching bisimulation equivalence class contains an effective process graph. We proved in \cite{BLT13} that a process is effective if and only it is executable.
In this result, it is needed to abstract from divergencies (infinite sequences of inert $\tau$-steps). Also note that the process graph of a queue automaton is always boundedly branching, but the minimal element of its branching bisimulation equivalence class can be infinitely branching. This is the case for the queue automaton shown in Figure~\ref{infbranch}.

\begin{figure}[htb]
\begin{center}
\begin{tikzpicture}[->,>=stealth',node distance=4cm, node font=\footnotesize, state/.style={circle, draw, minimum size=.5cm,inner sep=0pt}]
  \node[state,initial,initial text={},initial where=left] (s0) {};
  \node[state] [right of=s0] (s1) {};
  \node[state,accepting] [right of=s1] (s2) {};
  
  \path[->]
  (s0) edge[in=315,out=225,loop]
         node[below] {$\begin{array}{c}\tau[*/1]\\
                                      \tau[1/\varepsilon] \end{array}$} (s0)
  (s0) edge node[above] {$a[1/\varepsilon]$} (s1)
 (s1) edge[in=315,out=225,loop]
         node[below] {$a[1/\varepsilon]$} (s1)
  (s1) edge node[above] {$\tau[\varepsilon/\varepsilon]$} (s2);                                    
\end{tikzpicture}
\end{center}
\caption{Queue automaton, of which the process graph is branching bisimilar to an infinitely branching process graph.}\label{infbranch}
\end{figure}
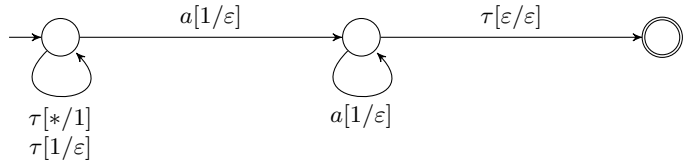

%In the next section we will introduce notations for executable processes.

 \section{Process Algebra}
 
 In this section, we express the executable processes in a mathematical formalism. In the theory of automata and formal languages, such a  notation is often called a  \emph{grammar}. In concurrency theory or process theory, such a notation is often called a \emph{process algebra}. In this process algebra, we can concisely express communication between executable processes.
 
 A queue automaton can be thought of as a finite automaton extended with queue memory. In this
section, we shall formalise this idea using process algebra, by proving that every executable process
can be specified as a regular process interacting with the queue process. Thereafter, we will also briefly
consider universal queue automata.
 
 In \cite{BBK85}, it was proven that every executable process can be defined as the solution of a finite recursive specification over the process algebra ACP$_{\tau}$; in \cite{BBR10}, this result was updated to the process algebra BCP$_{\tau}$ with standard communication. We proceed to present this process algebra here.
 
 \begin{defi}
 The syntax of BCP$_{\tau}$ with standard communication has the following ingredients:
 \begin{enumerate}
 \item there is a constant $\dl$ denoting inaction or deadlock: it denotes the finite automaton with a single state that is initial but not final and that has no transitions;
 \item there is a constant $\emp$ denoting termination or acceptance: it denotes the finite automaton with a single state that is initial and final and that has no transitions;
 \item there is a finite set of constants $\cal X$ of process identifiers or process variables; we use capital letters $X,Y, \ldots$ to range over $\cal X$;
 \item there is a set $\cal A$ of actions, a finite set $\cal D$ of data and a finite set $\cal C$ of communication ports: for each $d \in {\cal D}, c \in {\cal C}$ there are actions $c!d$ (send or output $d$ at $c$), $c?d$ (receive or input $d$ at $c$) and $c(d)$ (communicate $d$ at $c$); in addition, there is the unobservable action $\tau$; for each action $a \in {\cal A}_{\tau}$, there is the unary action prefix $a.\_$;
 \item there is the binary operator $+$ denoting choice or alternative composition;
 \item there is the binary operator $\merge$ denoting merge or parallel composition;
 \item for each communication port $c$, there is the unary operator $\encap{c}{\_}$ denoting restriction or encapsulation of all actions $c?d,c!d$ for $d \in {\cal D}$;
  \item there is the unary operator $\abstr{\cal C}{\_}$ denoting hiding or abstraction of all actions $c(d)$ for $d \in {\cal D},c \in {\cal C}$.
 \end{enumerate}
 \end{defi}
 
 A recursive specification over BCP$_{\tau}$ is a mapping
$\Gamma$ from $\cal X$ to the set of BCP$_{\tau}$ expressions. The idea is that the process expression
$p$ associated with a process identifier $X\in{\cal X}$ by $\Gamma$
\emph{defines} the behaviour of $X$. We prefer to think of $\Gamma$ as a
collection of \emph{defining equations}
  $X\defeqn p$,
  exactly one for every $X\in{\cal X}$.
We shall, in the sequel, presuppose a recursive specification
$\Gamma$ defining the process identifiers in $\cal X$, and we shall
usually simply write $X\defeqn p$ for $\Gamma(X)=p$. Note that, by our
assumption that $\cal X$ is finite, $\Gamma$ is finite too.

\begin{figure}[htb]
  \centering
  \begin{osrules}
        \osrule*{}{\emp \downarrow}
        \qquad \qquad
        \osrule*{}{\pref[a]p\step{a}p}
    \\
\osrule*{p \downarrow}{(p+q) \downarrow}
\quad
\osrule*{q \downarrow}{(p+q) \downarrow}
\quad
    \osrule*{p \step{a} p'}{p + q \step{a} p'}
    \quad
    \osrule*{q \step{a} q'}{p + q \step{a} q'}
  \\
 \osrule*{p \downarrow & q \downarrow}{p \merge  q \downarrow}
  \qquad
    \osrule*{p\step{a} p'}{p \merge  q \step{a} p' \merge  q}
  \qquad
    \osrule*{q \step{a} q' }{p \merge  q
      \step{a} p \merge q'}
\\
\osrule*{p \step{c!d} p' & q \step{c?d} q'}{p \merge  q \step{c(d)} p' \merge q'} \qquad
\osrule*{p \step{c?d} p' & q \step{c!d} q'}{p \merge  q \step{c(d)} p' \merge q'}
 \\
\osrule*{p \downarrow}{\encap{c}{p} \downarrow} \qquad \osrule*{p \step{a} p' & a \neq c!d,c?d}{\encap{c}{p} \step{a} \encap{c}{p'}} 
\\
\osrule*{p \downarrow}{\abstr{\cal C}{p} \downarrow}  \qquad \osrule*{p \step{c(d)} p'}{\abstr{\cal C}{p} \step{\tau} \abstr{\cal C}{p'}} \qquad \osrule*{p \step{a} p' & a \neq c(d)}{\abstr{\cal C}{p} \step{a} \abstr{\cal C}{p'}}
\\
 \osrule*{\term{p} & X\defeqn p}{\term{X}}  \qquad   \osrule*{p\step{a}p' & X\defeqn p}{X\step{a}p'}
  \end{osrules}
\caption{Operational semantics for BCP$_{\tau}$ with standard communication ($a \in \mathcal{A}_{\tau}, c \in \mathcal{C}, d \in \mathcal{D}, X \in \mathcal{X}$).}
\label{fig:semantics-tspseq}
\end{figure}

We associate behaviour with process expressions by defining, on the
set of process expressions, a unary acceptance predicate $\term{}$
(written postfix) and, for every $a\in\Act_{\tau}$, a binary transition
relation $\step{a}$ (written infix), by means of the transition system
specification presented in Figure~\ref{fig:semantics-tspseq}. 

By means of these rules, the set of process expressions turns into a labelled transition system, so we have strong bisimilarity and branching bisimilarity on  process expressions. 
Suppose $p$ is a BCP$_{\tau}$ expression (possibly containing variables from $\cal X$). Then the \emph{process} of $p$ is the branching bisimulation equivalence class of the process graph generated by the operational rules.

We explicitly state the result in \cite{BBR10}:

\begin{thm}
The process of every BCP$_{\tau}$ expression is executable. For every executable process, there is a BCP$_{\tau}$ expression with this process.
\end{thm}

\begin{cor}
Let $p,q$ be two BCP$_{\tau}$ expressions, and $c \in {\cal C}$ a communication port. Then $p \merge q$, $\encap{c}{p \merge q}$ and $\abstr{\cal C}{\encap{c}{p \merge q}}$ denote executable processes.
\end{cor}

In this way, we can express the communication of two processes. Using Theorem~\ref{2qtoq}, we can also obtain this result directly for queue automata.

%\begin{thm}
%If processes $p$ and $q$ are executable, connected by communication port $c$, then $\encap{c}{p \merge q}$ is also an executable process.
%\end{thm}
%\begin{proof}
%Since $p,q$ are executable, they are given by queue automata $M_p,M_q$. Then it is straightforward to give a queue automaton with two queues for $\encap{c}{p \merge q}$: the states will be pairs of states of the queue automata of $p,q$, each having their own memory contents, and the queues interact as follows. If in $M_p$ there is a transition $s \xrightarrow{c!d[d/\varepsilon]} s'$ and in $M_q$ there is a transition $t \xrightarrow{c?d[*/d]} t'$, then in $\encap{c}{p \merge q}$ there is a transition $(s,t) \xrightarrow{c(d)[(d,*),(\varepsilon, d)]} (s',t')$. Likewise, if in $M_q$ there is a transition $t \xrightarrow{c!d[d/\varepsilon]} t'$ and in $M_p$ there is a transition $s \xrightarrow{c?d[*/d]} s'$, then in $\encap{c}{p \merge q}$ there is a transition $(s,t) \xrightarrow{c(d)[(*,d),(d,\varepsilon)]} (s',t')$. Applying Theorem~\ref{2qtoq} yields the executability of $\encap{c}{p \merge q}$.
%\end{proof}

\begin{exa}
We give the following finite specification of the queue process of Example~\ref{queueqa}, adapted from \cite{GV93}:
\[ Q^{io} \defeqn \emp + o!\varepsilon.Q^{io} + \sum_{d \in {\cal D}} i?d.\abstr{\cal C}{\encap{\ell}{Q^{i\ell} \merge (\emp + o!d.Q^{\ell o})}} \]
This specification has 6 variables $Q^{io}, Q^{i\ell}, Q^{\ell o}, Q^{o\ell}, Q^{\ell i}, Q^{oi}$ and uses data set ${\cal D} \cup \{\varepsilon\}$.
\end{exa}

Recall that a regular process has a process graph with finitely many states and transitions.

\begin{thm}
Let $p$ be an executable process. Then there is a regular process $q$ such that $p \bbisim \abstr{\cal C}{\encap{io}{q \merge Q^{io}}}$.
\end{thm}
\begin{proof}
Let $p$ be an executable process. Then there is a queue automaton ${Q} = ({\cal S},{\cal A},{\cal D},\rightarrow,\uparrow,\downarrow)$ that defines $p$. By Lemma~\ref{singenqsepdeq}, we can assume without loss of generality that $Q$ has singleton enqueues and separate dequeues. In order to define $q$, we have to remember the head of the queue or remember that the queue is empty, in order to be able to know which following step is possible.
\begin{enumerate}
\item for each state $s \in {\cal S}$ and $d \in {\cal D} \cup \{\varepsilon\}$, $q$ has states $s_d, s_d^1, s_d^2, s_d^3$;
\item the queue $Q^{io}$ uses data set ${\cal D} \cup \{\$\}$; 
\item the initial state of $q$ is $\uparrow_{\varepsilon}$, and state $s_d$ is final whenever $s$ is final in $Q$;
\item whenever $Q$ has $s \xrightarrow{a[\varepsilon/\varepsilon]} t$, then $q$ has $s_{\varepsilon} \step{a} t_{\varepsilon}$;
\item whenever $Q$ has $s \xrightarrow{a[*/d]} t$ for some $d \in {\cal D}$, then $q$ has $s_e \step{a} t_e^1 \step{i!d} t_e$ for all $e \in {\cal D}$ and $s_{\varepsilon} \step{a} t_d^1 \step{i!d} t_d$;
\item whenever $Q$ has $s \xrightarrow{a[d/\varepsilon]} t$ for some $d \in {\cal D}$, then $q$ has $s_d \step{a} s_d^1 \step{o?d} s_d^2 \step{o?\varepsilon} t_{\varepsilon}$; moreover, for all $e \in {\cal D}$ $q$ has steps $s_d^2 \step{o?e} t_e^3 \step{i!\$} t_e^2 \step{i!e} t_e^1 \step{o?\$} t_e$ and for all $f \in {\cal D}$ steps $t_e^1 \step{o?f} t_e^1$.
\end{enumerate}
All input and output steps in $q$ will successfully communicate with the queue, and be turned into $\tau$-steps by the abstraction operator. All will turn out to be inert.
\end{proof}

Note that the converse of this theorem is also true, as every process defined by a finite specification over BCP$_{\tau}$ with standard communication is executable.

Just like we did in \cite{BLT13}, we can prove a universal queue automaton $U$ exists. For an arbitrary queue automaton $M$, let $\overline{M}$ be the deterministic queue automaton that outputs the G\"{o}del number of $M$ (in some appropriate representation) along a special communication port $u$, then terminates with empty queue, and has no other behaviour. Then a \emph{universal} queue automaton $U$ is such that $\abstr{\cal C}{\encap{u}{\overline{M} \merge U}}$ has the same process as $M$ for all queue automata $M$ (here, we use the BCP$_{\tau}$-expressions of these queue automata to define this process).

\section{A hierarchy}\label{characterisation}

A computer shows interaction between the finite control and the memory. The finite control can be represented by a regular process (a finite automaton). 
In this article, we considered a memory in the form of a queue.

In \cite{BCT08}, we considered a memory in the form of a stack, and we established that a pushdown process can be characterised as a regular process communicating with a stack. This work was continued in \cite{BCL23} and \cite{BL24} to find the process algebra TSP;$_{sc}$ that is associated with pushdown automata.  TSP;$_{sc}$ is obtained from the process algebra BCP$_{\tau}$ of the previous section by leaving out parallel composition, encapsulation and abstraction, and adding sequencing with sequential value passing.

In \cite{BL23}, we considered a memory in the form of a bag, and we established that a parallel pushdown process can be characterised as a regular process communicating with a bag. We found the process algebra associated with parallel pushdown automata. This process algebra is obtained from the process algebra BCP$_{\tau}$ of the previous section by leaving out the abstraction operator and adding the priority operator of \cite{BBK86}.

We see that the queue is the prototypical executable process, as all executable processes can be realised as a regular process communicating with a queue. Likewise,
 the bag is the prototypical parallel pushdown process. A bag is an executable process, but not a pushdown process. Further,
the stack is the prototypical pushdown process, but not a parallel
pushdown process. 

A counter can be realised as a stack with a singleton data set, and also as a bag with a singleton data set. Thus, it is in the intersection of pushdown processes and parallel pushdown processes. It is not a regular process, as it has infinitely many different states that are not bisimilar. We conjecture that every process in the intersection of pushdown processes and parallel pushdown processes can be realised as a regular process communicating with a counter, but have no proof of this yet.
 Figure~\ref{fig:classification} provides a complete picture.

\begin{figure}[htb]
\begin{center}
\begin{tikzpicture}
\usetikzlibrary{shapes}
\clip (0,0) circle (3.3cm);
\node[circle,minimum size=6.6cm,draw] (exe) at (0,0){};
\path (exe.north) coordinate (exetext) node[below,yshift=-0.2cm]{Executable};
\node[circle,minimum size=2.2cm,draw, anchor=south] (reg) at (exe.south){Regular};
\node[ellipse, minimum width=5.61cm, minimum height=3.52cm,draw,anchor=west,rotate=45] (ppd) at (reg.south west){};
\node[circle,minimum size=1.1cm,anchor=east,yshift=.3cm] (ppdtext) at
(exe.east){Parallel Pushdown};
\node[ellipse, minimum width=5.61cm, minimum height=3.52cm,draw,anchor=east,rotate=-45] (pd) at (reg.south east){};
\node[circle,minimum size=1.1cm,anchor=west,yshift=.3cm] (pdtext) at (exe.west){Pushdown};
\path (0,-0.5) coordinate (C) node[below]{Counter};
\fill (C) circle(2pt);
\path (2,-0.5) coordinate (B) node[below]{Bag};
\fill (B) circle(2pt);
\path (-2,-0.5) coordinate (S) node[below]{Stack};
\fill (S) circle(2pt);
\path (0,1.5) coordinate (Q) node[below]{Queue};
\fill (Q) circle(2pt);
\end{tikzpicture}
\end{center}
  \caption{Classification of Executable, Pushdown, Parallel Pushdown,
    and Regular processes and the prototypical processes Queue, Bag,
    Stack and Counter.}\label{fig:classification}
\end{figure}
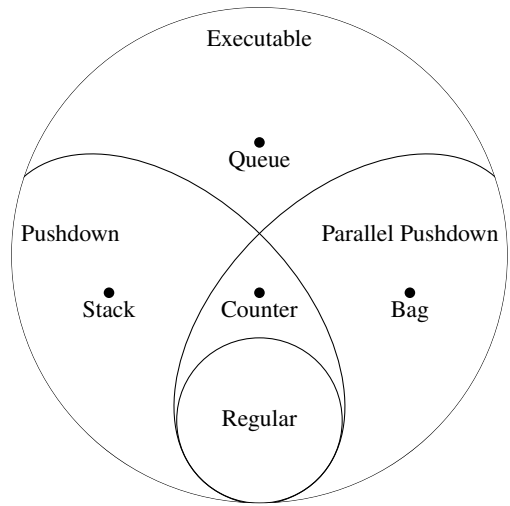

\section{Conclusion}

We considered the computational model of the Queue Automaton, and have proved that it is equally expressive as the Reactive Turing Machine of \cite{BLT13}. Thus, a process is executable if and only if it is the process of a queue automaton, and a language or function is computable if and only if it is the language or function of a queue automaton.
Every executable process can be defined as a regular process communicating with a queue. This fits in very well with earlier results, that every pushdown process can be defined as a regular process communicating with a stack, and a parallel pushdown process can be defined as a regular process communicating with a bag.
We think that a pushdown automaton can be better called a stack automaton, and a parallel pushdown automaton a bag automaton, in order to emphasise the relation to the queue automaton.
As grammar for executable processes we use the process algebra BCP$_{\tau}$, Basic Communicating Processes with abstraction and standard communication.

%%
%% Bibliography
%%

%% Please use bibtex, 

\bibliographystyle{eptcs.bst}
\bibliography{main}

\end{document}